\documentclass{article}[11pt,letter]
\pdfoutput=1

\usepackage[english]{babel}
\usepackage{graphicx}
\usepackage{epstopdf}
\usepackage[cmex10]{amsmath}
\usepackage{amsmath,amssymb}
\usepackage{amsthm}
\usepackage{mathtools}
\usepackage[utf8]{inputenc}
\usepackage{authblk} 
\usepackage[breaklinks,bookmarks=false]{hyperref}
\usepackage{algpseudocode}
\usepackage{algorithm}
\usepackage{xspace}
\usepackage{array}
\usepackage{cite}
\usepackage{microtype}
\usepackage{fullpage}

\usepackage{thm-restate}
\usepackage{thmtools}

\newcolumntype{R}[1]{>{\raggedleft\let\newline\\\arraybackslash\hspace{0pt}}p{#1}}

\usepackage{fixltx2e}
\MakeRobust{\Call}

\newtheorem{definition}{Definition}[section]
\newtheorem{lemma}[definition]{Lemma}
\newtheorem{theorem}[definition]{Theorem}
\newtheorem{corollary}[definition]{Corollary}
\newtheorem{proposition}[definition]{Proposition}

\newtheorem{observation}[definition]{Observation}

\usepackage{setspace}   
\usepackage{color}

\newcommand{\comp}[2]{\ensuremath{\widehat{#1[#2]}}}
\newcommand{\dom}[2]{\ensuremath{\mathsf{dom}_{#1}(#2)}}
\newcommand{\parent}[2]{\ensuremath{\mathsf{parent}_{#1}(#2)}}
\newcommand{\imm}[2]{\ensuremath{\mathsf{imm}_{#1}(#2)}}
\newcommand{\var}{\ensuremath{\text{var}}}
\newcommand{\vdeg}{\ensuremath{\text{deg}}}
\newcommand{\out}{\ensuremath{\text{in}}}
\newcommand{\arb}{\ensuremath{\mathsf{arb}}}
\newcommand{\lab}[1]{#1}
\newcommand{\bigo}{\mathcal{O}}

\newcommand{\eg}{e.g.\xspace}
\newcommand{\etal}{et~al.}
\newcommand{\aux}{\ensuremath{\mathsf{aux}}}
\newcommand{\rooot}{\ensuremath{\mathsf{root}}}

\newcommand{\variable}[1]{\ensuremath{\mathsf{#1}}}

\begin{document}



\author[1]{Mat\'u\v{s} Mihal\'{a}k}
\author[2]{Przemys\l{}aw Uznański}
\author[3,4]{Pencho Yordanov}

\affil[1]{Department of Knowledge Engineering, Maastricht University, The Netherlands}
\affil[2]{Helsinki Institute for Information Technology (HIIT), Aalto University, Finland}
\affil[3]{Department of Biosystems Science and Engineering, ETH Zurich, Switzerland}
\affil[4]{Swiss Institute of Bioinformatics, Basel, Switzerland }

\title{Prime Factorization of the Kirchhoff Polynomial:\\ Compact Enumeration of Arborescences}

\maketitle

\begin{abstract}

  We study the problem of enumerating all rooted directed spanning trees
  (arborescences) of a directed graph (digraph) $G=(V,E)$ of $n$ vertices. 
  An arborescence $A$ consisting of edges $e_1,\ldots,e_{n-1}$ can be
  represented as a monomial $\lab{e_1}\cdot\lab{e_2}\cdots\lab{e_{n-1}}$
  in variables 
  $\lab{e} \in E$.
  All arborescences $\arb(G)$ of a digraph then define the 
  Kirchhoff polynomial $\sum_{A \in \arb(G)} \prod_{e\in A} \lab{e}$.
  We show how to compute a compact representation of the Kirchhoff
  polynomial -- its prime factorization, and how it relates to combinatorial properties of
  digraphs such as strong connectivity and vertex domination. In particular, we
  provide digraph decomposition rules that correspond to factorization steps of
  the polynomial, and also give necessary and sufficient primality conditions
  of the resulting factors expressed by connectivity properties of the
  corresponding decomposed components.
  Thereby, we obtain a linear time algorithm for decomposing a digraph
  into components corresponding to factors of the initial polynomial, and a
  guarantee that no finer factorization is possible.
  The decomposition serves as a starting point for a recursive deletion-contraction algorithm,
  and also as a preprocessing phase for iterative enumeration algorithms. Both
  approaches produce a compressed output and retain some structural properties in
  the resulting polynomial. This proves advantageous in practical applications
  such as calculating steady states on digraphs governed by Laplacian dynamics,
  or computing the greatest common divisor of Kirchhoff polynomials.
  Finally, we initiate the study of a class of digraphs which allow for a
  practical enumeration of arborescences. Using our decomposition rules
  we observe that various digraphs from real-world applications fall into this
  class or are structurally similar to it.

\end{abstract}
\section{Introduction and Related Work}
A spanning tree of an undirected graph $G$ is a connected acyclic subgraph
containing all vertices of $G$. In a directed graph (digraph for short)
$G$, the analogue is an \emph{arborescence}, i.e., a subdigraph of $G$
spanning its vertex set such that all vertices are reachable from a root
vertex along a unique directed path.

We are interested in the problem of enumerating all arborescences of a digraph $G$. 
The famous Tutte's Matrix-Tree Theorem~\cite{tutte1948dissection} states that
the number of arborescences rooted at vertex $j$ in a digraph is equal,
up to a sign, to the $(jj)$th minor of the Kirchhoff matrix; the Kirchhoff
matrix, which is closely related to the Laplacian matrix, is the difference
between the diagonal matrix having the in-degree of the vertices on its diagonal
and the adjacency matrix of the digraph. (This theorem is itself a
generalization of the equally famous Kirchhoff's Theorem~\cite{Kirchhoff/1847}
for undirected graphs.)
Thus, \emph{counting} the number of arborescences can be achieved in
polynomial time (e.g., by the Gaussian elimination).
Arborescences can also be \emph{enumerated} with the Matrix-Tree
Theorem by uniquely labeling the edges of the digraph. Then the modified
Kirchhoff matrix is the difference between the diagonal matrix having the sum of
labels of all in-coming edges to each vertex on its diagonal and the adjacency
matrix constructed from the labels of edges connecting adjacent vertices.
Summing up all $(jj)$th minors results in a homogeneous polynomial, called the
Kirchhoff polynomial, in which each monomial represents an arborescence
consisting of the edges corresponding to the labels appearing in the monomial.
Formally, the Kirchhoff polynomial $\kappa(G)$ of a digraph $G$ over variables
$\lab{e} \in E$, is defined as
\begin{displaymath}
  \kappa(G) = \sum_{A \in \arb(G)} \prod_{e\in A} \lab{e}\text{,}
  \vspace{-1mm}
\end{displaymath}
where $\arb(G)$ denotes the set of all arborescences of $G$.
In general, there might be exponentially many arborescences for a given $G$, and their
enumeration by explicitly computing the canonical form of $\kappa(G)$
cannot be done in polynomial time.

Gabow and Myers presented an algorithm for enumerating all
arborescences with $O(|E| + N\cdot|E|)$ running time ($N$ is the number of
arborescences) and $O(|E|)$ space
requirements~\cite{gabow1978finding}.
Later, this algorithm has been improved by the currently two state-of-the-art
algorithms~\cite{unoa96,kk00}. Both algorithms start by
computing an initial (arbitrary) arborescence, and then, iteratively,
compute ``close-by'' arborescences by outputting only the edge-difference to the
previously computed and listed arborescence. Uno \cite{unoa96} uses a
reverse search, while Kapoor and Ramesh \cite{kk00} use depth-first search in
the space of all arborescences (represented as an undirected graph
where a node corresponds to an arborescence, and an edge denotes a
single edge-swap between the two adjacent arborescences). The algorithm of Uno
runs in $O(|E| + N\cdot \sqrt{|V|}\log(|E|/|V|))$ time and has space complexity
$O(|E|)$, and the algorithm of Kapoor and Ramesh runs in $O(N|V|+|V|^3)$ time
and has a space complexity $O(|V|^2)$.

In this paper we propose an alternative approach to the above algorithms,
focusing on the form of the output, which in our case is predominantly more
compressed, and allows for certain symbolic manipulations. 
Our enumeration algorithm is based on a (recursive) factorization of the
Kirchhoff polynomial $\kappa(G)$.
We present two decomposition rules for digraphs, and show that \emph{every}
factor of $\kappa(G)$ corresponds to a digraph obtained by applying one of the decomposition rules and is derived from some subdigraph of $G$. The first decomposition rule
corresponds to finding a strongly connected component (SCC) of $G$ together with
all of its in-coming edges $e_{in}$, and subsequently contracting all source
vertices of $e_{in}$ into a single new vertex.
The second decomposition rule is more involved and is based on \emph{dominance}
and \emph{immediate dominance} of vertices (these terms are explained in
Section~\ref{sec:preliminaries}).
Our decomposition is closest to the decomposition that has been formulated in the language of principal partition of matroids~\cite{nma80}.

To the best of our knowledge, the question of relating combinatorial features of a digraph to the algebraic properties of its Kirchhoff polynomial (such as its prime factorisation), although being a fundamental one, has not previously been investigated.

We illustrate the compressing potential of our method on an example in 
Figure~\ref{fig:compression_example}.
\begin{figure}[t]
  \centering
    \includegraphics[scale=1.5]{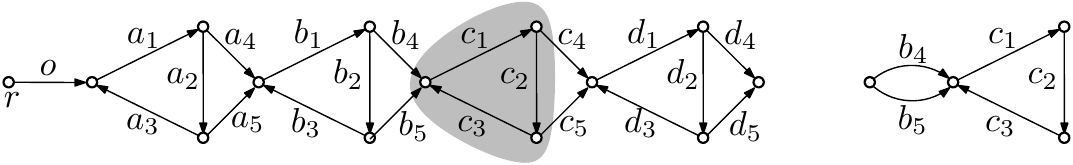}
    \caption{Left: A digraph with one trivial (the initial vertex) and four
      non-trivial strongly connected components (here cycles), each induced by
      the edges $\{x_1,x_2,x_3\}$, $x\in \{a,b,c,d\}$. The gray area depicts one
      such strongly connected component $G[V']$. Right: An illustration of
      $\comp{G}{V'}$.}
    \label{fig:compression_example}
\end{figure}
Obviously, every arborescence of the example is rooted in $r$, contains
the edge with label $o$, and contains the edges of type $x_1$ and $x_2$ from
every SCC. The only freedom left is choosing whether edge $x_4$ or edge $x_5$ is
part of an arborescence.
These choices are mutually independent, and as a result, there are $2^4$
arborescences, which results in $208$ explicitly written variables. The
factorized representation of the Kirchhoff polynomial (based on the SCC
decomposition) is much shorter:
$$o(a_1a_2a_4+a_1a_2a_5)(b_1b_2b_4+b_1b_2b_5)(c_1c_2c_4+c_1c_2c_5)(d_1d_2d_4+d_1d_2d_5),$$
which contains 25 explicitly written variables. We can further decompose
each factor to finally obtain the prime factorization
$$oa_1a_2(a_4+a_5)b_1b_2(b_4+b_5)c_1c_2(c_4+c_5)d_1d_2(d_4+d_5),$$
containing only 17 explicitly written variables.

An exhaustive application of the decomposition rules results in non-decomposable
digraphs derived from subdigraphs of $G$, which correspond to prime factors of $\kappa(G)$. The prime decomposition provides a compressed form of
$\kappa(G)$ which can be easily handled (evaluated) and manipulated
(e.g., finding the greatest common divisor of two Kirchhoff polynomials).
Moreover, the prime factorization of $\kappa(G)$ retains important
connectivity properties of the original digraph and provides information about
the type of digraphs with practically enumerable arborescences, which is
beneficial for various practical applications.
Furthermore, the decomposition/prime factorization can serve as a preprocessing
step for the existing iterative enumeration algorithms (such as those by Gabow
and Myers~\cite{gabow1978finding}, Uno~\cite{unoa96}, or Kapoor and
Ramesh~\cite{kk00}), but can also be a starting point for a recursive algorithm,
which employs the deletion-contraction step whenever further factorization
is impossible.
Finding the best edges (i.e., the edges that lead to the largest possible
compression) on which the recursion shall proceed is an interesting open problem. 
For our experimental evaluation, we resort to several simple heuristics that
can be of practical use, and that served well for demonstrating
the potential of our method.

The rest of the paper is organized as follows. Section~\ref{sec:preliminaries}
gives necessary terminology, definitions, and initial observations.
Section~\ref{sec:primality} characterizes the non-decomposability, in terms of the structural properties
of the input digraph, corresponding to primality of Kirchhoff polynomials.
Section~\ref{sec:decomposition} lays out the prime decomposition rules.
Section~\ref{sec:applications} discusses the applicability of the prime
decomposition and compressed enumeration, while demonstrating the potential of
our method by applying the recursive algorithm to real-world examples.

\section{Preliminaries}
\label{sec:preliminaries}

We are dealing with simple directed graphs (\emph{digraphs} for short), however all our
results also hold for digraphs with parallel edges (loops make no sense in our
applications, since no arborescence can ever contain one).
When necessary, we denote the vertices and the edges of a digraph $G$ by
$V(G)$ and $E(G)$, respectively.
For $V'\subseteq V$, $G[V']$ denotes the \emph{induced subdigraph} of $G$ by
the set of vertices $V'$.
We denote the set of all incoming edges to vertex $v$ by $\out_G(v)$ and
when it is not ambiguous, simply by $\out(v)$.

A digraph $G$ is strongly connected if for any two vertices $u,v \in V$ there is a
directed path from $u$ to $v$ and from $v$ to $u$.
A strongly connected component (SCC for short) of $G$ is any largest (w.r.t.
vertex inclusion) strongly connected induced subdigraph of $G$. 
It follows that no two distinct SCCs can share a
vertex, and, therefore, all strongly connected components $G_1,\ldots,G_k$ of a
digraph $G$ induce a unique partition $V(G_1),\ldots,V(G_k)$ of $V$.

Observe that for two distinct SCCs $G_i$ and $G_j$ there can be a directed
path from $G_i$ to $G_j$, or a path from $G_j$ to $G_i$, but not both. The
existence of such paths between SCCs naturally induces a unique partial order on
the SCCs $G_1,\ldots,G_k$.

Recall that an arborescence $A$ is a rooted directed spanning tree of $G$ with
edges directed away from the root. Observe that in an arborescence $A$,
the root is the unique vertex that can reach every other vertex along a directed path
in $A$.
We denote by $\arb(G)$ the set of all arborescences of $G$, and by
$\arb_v(G)$ the set of all arborescences rooted at vertex $v$.
Let $\rooot_v(G)$ be the digraph constructed from $G$ by removing all edges
incoming to $v$. All arborescences of $\rooot_v(G)$ are then
necessarily rooted at $v$. We say that $G$ is \emph{rooted} at vertex $v$,
if $v$ has no incoming edges, and every other vertex is reachable from $v$.
Observe that every arborescence of a digraph rooted at a certain vertex $v$ is itself also a digraph rooted at vertex $v$, so we find it justified to use the same term ``rooted'' when referring to digraphs and to arborescences.

We represent the set of all arborescences in a digraph $G$ with a
homogeneous multivariate polynomial over the variables $\lab{e}$, $e\in E$,
called the \emph{Kirchhoff polynomial} (or synonymously, the
arborescence \emph{enumerator}) $\kappa(G)$. For such a multivariate polynomial
$P$, we denote by $\text{var}(P)$ the set of variables appearing in $P$ (in
monomials with non-zero coefficients).

Besides the Kirchhoff polynomial $\kappa(G) = \sum_{A \in \arb(G)} \prod_{e\in
E(A)} \lab{e}$, we sometimes are only interested in arborescences rooted in
a specific vertex $v$. Then we consider the related polynomial $\kappa_v(G)
= \sum_{A \in \arb_v(G)} \prod_{e   \in E(A)} \lab{e}$.\footnote{Please, note that for a singleton $G$ (that is $V(G) = \{v\}$) the \emph{empty} digraph is its only spanning arborescence, thus formally $\kappa(G) = 1$.}

A polynomial $P$ is a \emph{factor} of a polynomial $Q$, if there exists a polynomial
$R$ such that $Q = P\cdot R$. A polynomial $P$ that has only trivial factors is
called \emph{prime}. Similarly, we say that \emph{$G'$ is a component (a prime
component) of $G$} if \emph{$\kappa(G')$ is a factor (a prime factor) of
$\kappa(G)$}.

In general, results presented in this paper hold even when considering
multidigraphs instead of simple directed graphs.
Through our digraph manipulations we could sometimes obtain multidigraphs
(with loops) as a result. However, this does not present a problem, as we can
transform the multidigraph into a simple digraph whose Kirchhoff polynomial
is equal to that of the multidigraph:
\begin{observation}
  The following two operations on a multidigraph $G$ preserve $\kappa(G)$:
\begin{enumerate}
\item Removal of all loops,
\item Replacing multiple parallel directed edges $e_1,e_2,\ldots,e_i$ going from
  $u$ to $v$ with a single edge from $u$ to $v$ such that $\lab{e} =
  \lab{e_1}+\lab{e_2}+\ldots+\lab{e_i}$.
\end{enumerate}
\end{observation}

\begin{observation}
  An arborescence of $G$ exists iff the partial order of the
  SCCs has exactly one minimal element. Such a SCC is called \emph{the initial
  SCC}.
\end{observation}

\begin{definition}
  Let $G[V']$ be a SCC of $G$. By $\comp{G}{V'}$ we
  denote the digraph created from $G[V']$ as follows (cf.\,
  Figure~\ref{fig:compression_example}):
  \begin{itemize}
    \item If $G[V']$ is the initial SCC, then
      $\comp{G}{V'} = G[V']$.
    \item Otherwise, we create a new vertex $v_{\aux}$, and for every edge $vu$,
      such that $u \in V'$ and $v \not\in V'$, we add an edge $v_{\aux}u$ with
      label $\lab{v_{\aux}u} = \lab{vu}$.
  \end{itemize}
\end{definition}

\begin{definition}[Domination]
  If $G$ is rooted at $v$, then we say that vertex $u$ dominates vertex $w$, if
  all directed paths from $v$ to $w$ go through vertex $u$. 
  By $\dom{G}{u}$ we denote the set of all vertices of $G$ dominated by $u$. 
  If $\dom{G}{u} \not= \{u\}$ and $u\not= v$, we say that $u$ is a
  non-trivial dominator.
\end{definition}

\begin{definition}[Immediate Domination]
  Let $G$ be rooted at $v$. Vertex $y$ is called an immediate dominator of
  vertex $z$, if $y \not= z$, $y$ dominates $z$, and for every other
  vertex $x$ that dominates $z$, we have that $x$ also dominates $y$. 
    Equivalently we say that $z$ is immediately dominated by $y$, and we denote such $y$ as $\parent{G}{z}$ (it is easy to see that there can be at most one such vertex).
  The set of all immediately dominated vertices by vertex $y$ is denoted by 
  $\imm{G}{y} = \{ z : \parent{G}{z} = y \}$.  
\end{definition}
Thus, $\parent{G}{z}$ taken over all $z \in V$ defines a directed tree $T(G)$ of
immediate domination, rooted at the root $v$ of $G$ (note that $\parent{G}{v}$
is undefined).
Furthermore, for any vertex $y$, the immediately dominated vertices $z\in
\imm{G}{y}$ induce the following partition of $\dom{G}{y}\setminus\{y\}$:
$\{\dom{G}{z}\,:\, z\in \imm{G}{y} \}$.

\begin{definition}[Contraction]
  For a subset $S$ of the vertices of a digraph $G$, and for a vertex $u \in S$,
  by contracted digraph $G(S \to u)$ we denote the digraph $G'$ constructed from $G$ as follows:
  \begin{enumerate}
    \item All edges $xy$, where $x \in V\setminus S$, $y \in S\setminus\{u\}$ are
      removed from $G$.
    \item All edges within $S$ are removed (i.e., all edges $xy$, where $x,y\in
      S$).
    \item All vertices of $S$ are contracted into a single vertex $u$.
  \end{enumerate}
\end{definition}

Observe that $G'$ has no loops, 
and if $G$ has no parallel edges incoming to $u$, then $G'$, as well, has no parallel edges incoming to $u$. It can, however, happen that $u$
has parallel outgoing edges.
Further, for all $z \in V \setminus S$:
\begin{itemize} 
  \item $zu \in E(G')$ if and only if $zu \in E(G)$, and
  \item for every edge $sz \in E(G)$ such that $s \in S$, there is a corresponding edge $uz \in E(G')$.
\end{itemize}

\section{Primality of Components}
\label{sec:primality}

Observe that $\kappa(G)$ is a special homogeneous polynomial: every monomial of
$\kappa(G)$ contains exactly $n-1$ variables, each with exponent equal to one.
Obviously, this property needs to also hold for factors of $\kappa(G)$ (recall
that the monomials of $\kappa(G)$ represent an arborescence of $G$).
Furthermore, observe that no variable $\lab{e}$ can appear in two factors of
$\kappa(G)$.

\begin{proposition}
  \label{prop:all_monomials_same_degree}
  If $P$ is a factor of $\kappa(G)$, then all monomials of $P$ have the same
  number of variables, each with exponent equal to one.
\end{proposition}

\begin{proposition}
  \label{lem:edg_part}
  If $\kappa(G) = P \cdot Q$, then $\var(P) \cap \var(Q) = \emptyset\text{.}$
\end{proposition}

We observe that the partitioning of edges into $P$ or $Q$ under factorization
$\kappa(G) = P\cdot Q$ is induced by a partitioning of vertices.

\begin{lemma}
  \label{lem:ver_part}
  If $\kappa(G) = P \cdot Q$ and $v \in V$, then either $\out(v) \subseteq \var(P)$ or $\out(v) \subseteq \var(Q)$.
\end{lemma}
\begin{proof}
Assume that there are two incoming edges to $v$, $e_1$ and $e_2$, such that $e_1
\in \var(P)$ and $e_2 \in \var(Q)$. Then there exists a monomial in $\kappa(G)$
(in its canonical form) containing both $e_1$ and $e_2$. But such a monomial
cannot represent an arborescence, a contradiction. 
\end{proof}

\begin{theorem}
  \label{th:nonfact1}
  Let $G$ be a strongly connected digraph. Then $\kappa(G)$ is prime.
\end{theorem}
\begin{proof}
Assume, on the contrary, that $\kappa(G) = P \cdot Q$, and $P$ and $Q$ are
nontrivial factors. Let $V_1$ and $V_2$ be the set of vertices with incoming
edges in $\var(P)$ and $\var(Q)$, respectively. By
Proposition~\ref{lem:edg_part} and Lemma~\ref{lem:ver_part}, $V_1 \cap V_2 =
\emptyset$, and since $P$ and $Q$ are nontrivial, $V_1 \not= \emptyset$ and $V_2
\not=\emptyset$. Let $v_1,v_2$ be arbitrarily picked vertices such that $v_1 \in
V_1$ and $v_2 \in V_2$.

Since $G$ is strongly connected, for any $v \in V$ there exists an
arborescence of $G$ rooted at $v$. Let $A_1,A_2$ be two
arborescences rooted at $v_1$ and $v_2$, respectively. Let $p_1$ and $p_2$ be
the monomials from $P$ corresponding to the arborescences $A_1$ and $A_2$,
respectively.
In $A_1$, every vertex from $V_1$, but the root $v_1$, has exactly one incoming
edge in $A_1$. The label of every such edge necessarily belongs to $P$.
Therefore, for $A_1$, the degree of the monomial $p_1$ in $P$ is $\vdeg(p_1) =
|V_1|-1$. On the other hand, for $A_2$, all vertices from $V_1$ have an incoming
edge whose label belongs to $P$, and therefore the degree of $p_2$ in $P$ is
$\vdeg(p_2) = |V_1|$. Then $\vdeg(p_1) \not= \vdeg(p_2)$, which contradicts
Proposition~\ref{prop:all_monomials_same_degree}.
\end{proof}

\begin{restatable}{theorem}{restatablenofacttwo}
  \label{th:nonfact2}
  Let $G$ be a digraph rooted at $v$, such that $G[V(G) \setminus \{v\}]$ is strongly
  connected, and $G$ has no non-trivial dominators. Then $\kappa(G)$ is prime.
\end{restatable}

To prove the above theorem, we use the following notion, and prove further
auxiliary lemmas. 
Given a directed simple cycle $\mathcal{C} \subseteq G$,
we say that a vertex $u \in \mathcal{C}$ is \emph{independent} from
$\mathcal{C}$, if there exists a simple path $\mathcal{P}$ connecting root $v$
to $u$, such that $\mathcal{P}$ and $\mathcal{C}$ are vertex-disjoint (except
for $u$). We call any such $\mathcal{P}$ an \emph{independent path} of $u$
(with respect to $\mathcal{C}$).

\begin{restatable}{lemma}{restatablegraphone}
  \label{lem:graph1}
  Let $G$ be a digraph as in Theorem \ref{th:nonfact2}.
  For any edge $wu$ of $G$ such that $u,w \not= v$, there exists a simple
  directed cycle $\mathcal{C}$ containing $wu$, such that $\mathcal{C}$ has at
  least two independent vertices, and $w$ is one of those.
\end{restatable}

\begin{restatable}{lemma}{restatablegraphtwo}
  \label{lem:graph2}
  Let $G$ be a digraph as in Theorem \ref{th:nonfact2}, and let $V_1$ and $V_2$ be
  an arbitrary (non-trivial) partition of the vertices $V \setminus \{v\}$.
  There exists a simple directed cycle having an independent vertex from $V_1$
  and an independent vertex from $V_2$.
\end{restatable}

\section{Decomposition}
\label{sec:decomposition}
In this section we present two digraph decomposition rules corresponding to factorization steps of the Kirchhoff polynomial. The rules are based on the computation of SCCs and the dominator tree of a digraph. The exhaustive application of these rules yields digraphs that are prime factors of the Kirchhoff polynomial of the original digraph.

\begin{restatable}{theorem}{restatablesccdecomposition}
\label{th:sccDecomposition}
Let $G[V_1],G[V_2],\ldots,G[V_k]$ be all strongly connected components of a connected digraph $G$.
If $G$ has exactly one initial component, then
\begin{equation}
\label{eq:scc_fact}
\kappa(G) = \kappa( \comp{G}{V_1} ) \cdot \kappa(\comp{G}{V_2}) \cdot \ldots \cdot \kappa(\comp{G}{V_k})\text{.}
\end{equation}

\end{restatable}

Figure~\ref{fig:SCCD} presents an example on how SCC decomposition is employed to factorise $\kappa(G)$.
The presented decomposition uncovers a fundamental property of arborescences. Namely, it shows that the arborescences of a digraph $G$ are in a one-to-one correspondence 
with all combinations of subdigraphs of $G$ obtained 
following the procedure: i) Pick an arborescence from the initial SCC of $G$. ii) For every non-initial SCC of $G$, pick as set $W$ an arbitrary (nonempty) subset of all vertices with incoming edges from outside of this SCC, and pick the spanning forest of this SCC rooted in $W$.

An equivalent formulation (used in the proof of Theorem~\ref{th:sccDecomposition}) is that any cycle can only be contained in a single SCC of $G$.
This property allows us to relate the Kirchhoff polynomial of a digraph to the product of Kirchhoff polynomials of digraphs derived from its SCCs.

\begin{figure}[t!]
\centering\includegraphics[width=\linewidth]{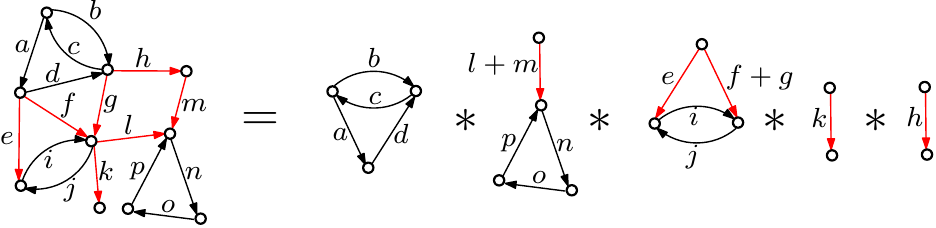}
\caption{Example of a digraph (left) and its decomposition using strongly connected components (right). The edges in black are part of a strongly connected component and the edges in red are connecting different SCCs.}
\label{fig:SCCD}
\end{figure}

Theorem~\ref{th:sccDecomposition} allows us to factorize $\kappa(G)$ of any connected digraph $G$ with at least two SCCs. Yet, it is not guaranteed that the obtained factorization is prime or non-trivial. 
Consider the case when the initial SCC is composed of a single vertex $v$, then the theorem states that $\kappa(G)=\kappa(G[\lbrace v \rbrace])\cdot\kappa(\comp{G}{V\setminus \lbrace v \rbrace})$. We can see that this is a trivial factorisation since $\kappa(G[\lbrace v \rbrace])=1$ 
and $\kappa(\comp{G}{V\setminus \lbrace v \rbrace})=\kappa(G)$ (here, $\comp{G}{V\setminus \lbrace v \rbrace}$ renames the vertex $v$ to $v_{aux}$ but preserves the arborescences).
In Theorem~\ref{th:nonfact1} we proved that the Kirchhoff polynomials of strongly connected digraphs are prime, which implies that the factor corresponding to the initial SCC is always prime. We note that the rest of the factors cannot be further non-trivially decomposed just using Theorem~\ref{th:sccDecomposition}. Their primality is unsettled because they lack the property of strongly connected digraphs, namely, that any vertex of the digraph is the root of an arborescence, due to possessing a single root for all arborescences (the auxiliary vertex). Thereby, we proceed to studying the decomposability of non-initial SCC factors.

With Theorem~\ref{th:domination_decomposition} we specify a decomposition step for non-initial SCC factors. More precisely, we provide an additional factorization rule of $\kappa(G)$ by using vertex domination relations (with respect to the root vertex $v_{aux}$). 

\begin{restatable}{theorem}{restatabledominationdecomposition}
\label{th:domination_decomposition}
Let $G$ be a digraph rooted at $v$ and let $u$ be an arbitrarily picked vertex of $G$. 
Denote $D = \dom{G}{u}$.
Then
\begin{equation}
\label{eq:kappa4}
\kappa(G) = \kappa( \rooot_u (G[D]) ) \cdot \kappa( G(D \to u ) ).
\end{equation}
\end{restatable}

Similarly to SCC decomposition, one can interpret this result as either a one-to-one correspondence between arborescences of $G$ and all combinations of the arborescences of its two factors, or as a statement on the structure of cyclic sets of edges in $G$ (any cycle is either a cycle when restricted to $D$, or remains a cycle when $D$ is contracted to $u$).

\begin{figure}[t!]
\centering\includegraphics[width=\linewidth]{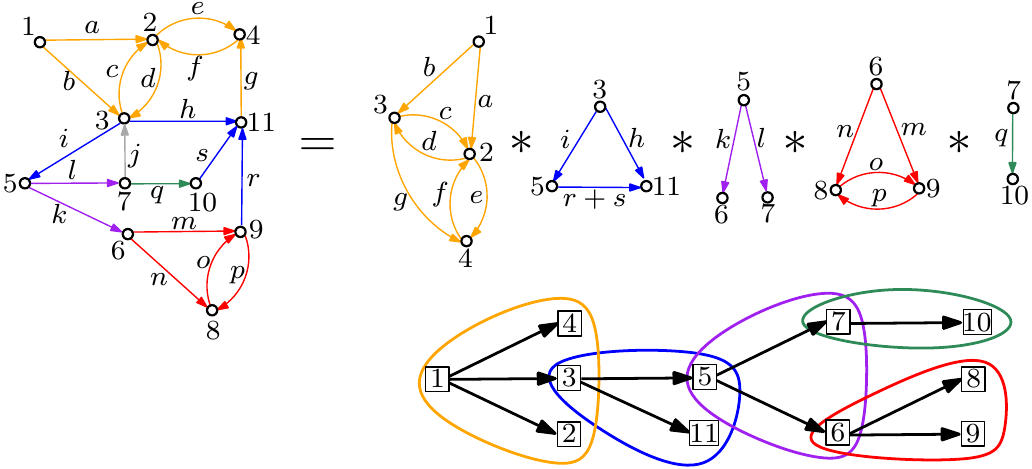}
\caption{Kirchhoff polynomial factorization of a digraph with respect to its dominator tree (above). The dominator tree of the example digraph (below). Color coding corresponds to immediate domination.}
\label{fig:Dom}
\end{figure}

Applying Theorem~\ref{th:domination_decomposition} to all domination relations defined by the dominator tree $T(G)$ extends the decomposition to the whole dominator tree at once. An illustration of the dominator decomposition rule can be seen in Figure~\ref{fig:Dom}.

\begin{corollary}
\label{cor:tree_decomp}
Let $G$ be a digraph rooted at $v$. Then
\begin{equation}
\label{eq:tree_decomp}
\kappa(G) = \prod_{u \in V} \kappa_u(H[\{u,u_1,\ldots,u_i\}] )
\end{equation}
where  $H = G( \dom{G}{u_1} \to u_1; \ldots ;\dom{G}{u_i}\to u_i)$ and $\{u_1,\ldots,u_i\} = \imm{G}{u}$.
\end{corollary}

Recall Theorem~\ref{th:nonfact2} which states that if there are no non-trivial domination relations in $G$ ($G$ being rooted at $v$ for which no non-trivial decomposition by SCCs applies) then $\kappa(G)$ is prime. Corollary~\ref{cor:tree_decomp} ensures that all non-trivial dominator relations are eliminated but that cannot guarantee the primality of the decomposed factors since they can be SCC decomposable. Therefore, in order to obtain a complete prime factorization of the input digraph we need to use both rules of decomposition in an alternating fashion. One might expect that this process requires deep recursion. However, for any digraph a constant depth of recursion is needed. More precisely, it is enough to apply in sequence SCC factorization, dominators factorization, and SCC factorization to get prime components. This upper bound on the recursion depth is tight, as the example in Figure~\ref{fig:counter} shows.

\begin{figure}[t!]
\centering\includegraphics[width=0.3\linewidth]{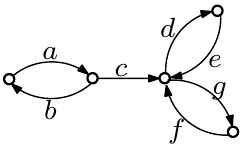}
\caption{Example of a digraph that requires SCC decomposition followed by dominators decomposition and another SCC decomposition to obtain prime factorization.}
\label{fig:counter}
\end{figure}

\begin{restatable}{theorem}{restatableththreelevels}
\label{th:3levels}
If $G$ is rooted at $v$, then any $H$ being a factor obtained by the application of rule \eqref{eq:tree_decomp} has the following property: all factors of $H$ obtained by applying \eqref{eq:scc_fact} are prime.
\end{restatable}

\paragraph{Pseudocode and time complexity.}
The pseudocode of the prime decomposition algorithm through SCCs and dominator relations can be found as Algorithm~\ref{alg:decomp} in Appendix~\ref{sec:pseudocode}.
Observe that obtaining the SCC decomposition takes time $\bigo(|V|+|E|)$ by Tarjan's strongly connected components algorithm \cite{tarjan1972depth}. Similarly, by the result of Alstrup \etal~\cite{Alstrup97dominatorsin}, one can find the dominator tree of a digraph in time $\bigo(|V|+|E|)$ .

\begin{corollary}
There is an algorithm that finds the decomposition of $G$ into prime components $G_1,\ldots,G_k$ in time $\bigo(|V|+|E|)$.
\end{corollary}

\section{Applications}
\label{sec:applications}

Our decomposition technique has numerous uses, owing to its properties coming from the prime factorization of the Kirchhoff polynomial, preservation of structural properties, and compression. 
Here we illustrate the fundamental insight provided by the decomposition, namely that
practical arborescence enumeration is not directly dependent on the exponentially growing number of arborescences but, rather, depends on digraph structure.
We show that this structure can be uncovered in real-life networks by means of appropriate edge deletion-contraction.
Further, we refer to several real-life applications that immediately benefit from the decomposition and the compact representation.

\paragraph*{Class of practically enumerable (PE) digraphs.}
For some digraphs, our decomposition algorithm results in a great speed-up compared to
standard techniques, and for some digraphs it does not help much (just think of
a complete digraph).
Our factorization rules, however, hint at what the structure of a digraph needs to
be, in order to allow an effective (practical) enumeration of
arborescences. 
An immediate example are digraphs, whose Kirchhoff polynomials factorize
exclusively to small (say, constant size) prime factors.
These digraphs (and others) form the class of \emph{practically enumerable
digraphs} -- the class $\text{PE}$ for short.
In this paper we do not aim to classify PE exactly.

We note that the mechanics of the SCC decomposition
(Theorem~\ref{th:sccDecomposition}) and domination decomposition
(Theorem~\ref{th:domination_decomposition}) can be reversed to define a
build-up procedure for generating non-trivial digraphs belonging to the
class PE. For example, given two rooted digraphs $P,G\in \text{PE}$, we can obtain
the digraph $P\circ G$ in which all vertices of $G$ are dominated by one
particular vertex $p$ from $P$ and $\kappa(P\circ G)=\kappa(P)\cdot\kappa(G)$.
Additionally, $G$ can be composed of multiple SCCs factorizing as
$\kappa(G)=\kappa(P_1)\cdot\ldots \cdot \kappa(P_n)$, where $P_1,\ldots,P_n\in
\text{PE}$. 
Interestingly, many digraphs obtained from real-life applications possess rich hierarchical and modular structure \cite{barabasi2004network,meunier2010modular}, which facilitates the practical enumeration of their arborescences.

\paragraph{Recursive enumeration.}
The classical edge deletion-contraction result allows to compute $\kappa(G)$
recursively even when further factorization is impossible:
\begin{observation}[Deletion-Contraction]
\label{obs:delcontr}
Let $G$ be a digraph and $uv \in E$ be an arbitrarily chosen edge. Then
$$\kappa(G) = \kappa(G \setminus \{uv\}) + \lab{uv} \cdot \kappa(G( \{u,v\} \to
u)).$$ 
\end{observation}

This helps us further exploit the prime decomposition by formulating a recursive
algorithm alternating between prime factorization and edge deletion-contraction
in every prime factor. It is an open problem which edges to delete-contract in
order to obtain a maximally compressed Kirchhoff polynomial. However, various
compression heuristics can be employed, \eg by choosing i) an edge whose
deletion produces the largest number of SCCs in the digraph (a strong bridge), ii)
an edge whose deletion introduces the most non-trivial domination relationships
among vertices, or iii) an edge whose contraction eliminates the greatest number
of edges.

We applied the recursive algorithm, using the heuristic of deletion-contraction
of the edge that maximizes the number of SCCs, to a wide spectrum of real-life and
randomly generated digraphs. Performance assessment of the recursive
algorithm is shown in Table~\ref{tab:networks}, and detailed description of the
digraphs can be found in Appendix~\ref{app:expgraphs}. We obtain small running
times and significant compression ratios ranging from $4.9$, for a small
biochemical network with $26$ arborescences, to $10^{372}$, for a large
randomly generated PE digraph with $2^{10}\cdot 3^{775}$ arborescences. The significant
compression indicates that the investigated digraphs are PE digraphs in
practice, and that the recursive algorithm performs well to uncover their 
structure. Our results further show that we can efficiently generate and handle
Kirchhoff polynomials of considerable size originating from practically relevant
digraphs.

\begin{table}[t!]
\centering
      \caption{Evaluation of the performance of the recursive algorithm on
    experimentally and randomly generated digraphs. Lengths of the compressed
    representation, $\vert C(\kappa(G)) \vert$, and of the expanded one, $\vert
    \kappa(G) \vert$, are measured in the number of symbols (edge labels and
    arithmetic operators).
    }
  \label{tab:networks}
  \begin{tabular}{lrrrrrr R{0.16\linewidth} R{0.05\linewidth}}
    \hline\noalign{\smallskip}
$G$& $\vert V\vert$& $\vert E\vert$& $\vert SCC \vert$& $|\arb(G)|$& $\vert \kappa(G) \vert$& $\vert C(\kappa(G)) \vert$& $ \vert \kappa(G) \vert / \vert C(\kappa(G)\vert$   \\
	\noalign{\smallskip}
	\hline
	\noalign{\smallskip}
ColE1&		$6$&		$10$&	$1$&	$26$&						$311$&							$63$&				$4.9$&				\\
Brain&		$8$&		$28$&	$8$&	$5,040$&					$80,639$&						$55$&				$1466$&				\\
MDH&		$9$&		$18$&	$1$&	$141$&						$2,537$&						$199$&				$12.7$&				\\
KNF33&		$9$&		$24$&	$1$&	$1,728$&					$31,103$&						$2,673$&			$11.6$&				\\
PHO5&		$12$&		$35$&	$1$&	$53,376$&					$1,281,023$&					$19,435$&			$65.9$&				\\
GR&			$13$&		$32$&	$1$&	$5,057$&					$131,481$&						$5,583$&			$23.5$&				\\
ERG&		$17$&		$55$&	$1$&	$44,741,862$&				$\sim 1.5\cdot 10^9$&			$134,483$&			$11,311$&			\\
SF1&		$18$&		$48$&	$1$&	$323,167$&					$11,634,011$&					$7,913$&			$1470$&				\\
SF3&		$40$&		$105$&	$1$&	$2,598,830,145$&			$\sim 2 \cdot 10^{11}$&			$815,411$&			$254,971$&			\\
S838&		$512$&		$819$&	$391$&	$\sim 6.2\cdot 10^{52}$&	$\sim 6.3\cdot 10^{55}$&		$1,515$&			$10^{52}$&			\\
PE45&		$1,561$&	$3,125$&$6$&	$2^{10}\cdot 3^{775}$& 		$\sim 2^{21} \cdot 3^{775}$&	$7,809$& 			$10^{372}$&			\\
    \hline
    \end{tabular} 
 
\end{table}

\paragraph*{Steady state on digraphs governed by Laplacian dynamics.}
A prototypical task for which arborescence enumeration is essential is
finding steady states on digraphs with Laplacian dynamics. Particular examples
include the symbolic derivation of kinetic equations~\cite{qi2009generating} and
steady state expressions for biochemical systems~\cite{gunawardena2012linear}.
However, the exponential growth of the number of arborescences with
digraph size has limited practical analyses only to small digraphs. In
Table~\ref{tab:networks} we show that the recursive algorithm performs well and
produces a compact output on various biochemical networks (described in
Table~\ref{tab:netref} in the appendix) that are regarded tedious to study
symbolically.

\paragraph*{Greatest common divisor of enumerators.}
The prime decomposition permits fast calculation of the greatest common divisor
(GCD) of Kirchhoff polynomials, bypassing costly symbolic manipulations of
multivariate polynomials. GCDs are required for simplifying ratios of Kirchhoff
polynomials and defining equivalence classes on them. This is an important tool
for the study of ratios of steady states on Laplacian systems and facilitates
the extension of the notion of spanning edge betweenness
\cite{teixeira2013spanning} to digraphs.

\paragraph*{Enumeration of specific arborescences.}
In many situations, one needs to enumerate/count only arborescences of certain
type, e.g., those of minimum cost, or those of bounded degree.
Our preliminary experimental investigation suggests that our decomposition
technique might be of practical/theoretical use. 

\paragraph*{Random sampling.}
Finally, note that we can sample, uniformly at random, an arborescence
from the compressed form of $\kappa(G)$ as follows: pick a random monomial from
each of the (also recursively computed) components.

\section{Conclusions}
\label{sec:conclusion}

We studied the fundamental problem of enumerating arborescences
via an explicit calculation of the Kirchhoff polynomial $\kappa(G)$ of a digraph
$G$. 
We provided a graph-theoretic structural characterization of digraphs corresponding
to prime factors of $\kappa(G)$. Based on this we presented a linear-time digraph decomposition technique that corresponds to
the prime factorization of $\kappa(G)$. Combining the decomposition with the
standard edge deletion-contraction recursion, we obtained a recursive algorithm
that further compresses the representation of $\kappa(G)$. 
We tested the algorithm on digraphs from real-world applications, and
obtained short running times and compact representations. 
\subsubsection*{Acknowledgements.}
We are grateful to Peter Widmayer and J{\"o}rg Stelling for the support and
fruitful discussions. We thank Victor Chepoi for valuable suggestions
regarding the proof of Lemma~\ref{lem:graph1}. PY was supported by FP7
Collaborative EEC project IFNaction. Part of this work was done while PU was
affiliated with Aix-Marseille Université and was visiting ETH Zurich, and while
MM was affiliated with ETH Zurich.

\bibliographystyle{splncs03}
\bibliography{refAlg}
\newpage

\pagenumbering{roman}

\appendix
\section{Pseudocode}
\label{sec:pseudocode}

\begin{algorithm}[h!]
\caption{Digraph decomposition corresponding to Kirchhoff polynomial prime factorization.}
\label{alg:decomp}
\begin{algorithmic}[1]
\Function{SCCFactors}{\variable{G}} 
	\State $\variable{Factors} \gets []$
	\ForAll{ $\variable{S} \in \Call{SCCs}{\variable{G}}$} \Comment{strongly connected components}
		\State $\variable{H} \gets \variable{G}[\variable{S}]$ \Comment{induced component}
		\If{ $\exists \variable{v}\variable{u} \in E(\variable{G}): \variable{v} \not\in \variable{S} \wedge \variable{u}\in \variable{S}$ } \Comment{non-initial SCC}
			\State $\variable{H}.\Call{addVertex}{\variable{v}_{\aux}}$
			\ForAll{ $\variable{v}\variable{u} \in E(\variable{G})$ : $\variable{v} \not\in \variable{S} \wedge \variable{u}\in \variable{S}$}
				\State $\variable{H}.\Call{addEdge}{\variable{v}_{\aux}\variable{u}}$
			\EndFor	
		\EndIf
		\State $\variable{Factors}.\Call{append}{\variable{H}}$
	\EndFor
	\State \Return $\variable{Factors}$
	
\EndFunction
\State {}
\Function{DominationFactors}{\variable{G}}
	\If{ $\neg \Call{isRooted}{\variable{G}}$ } \Comment{$G$ has to be rooted}
		\State \Return $[G]$
	\EndIf
	\State $\variable{T} \gets \Call{DominatorsTree}{\variable{G},\Call{getRoot}{\variable{G}}}$ 
	\State $\variable{Factors} \gets []$
	\State $\variable{H} \gets \variable{G}$ \Comment{copy of $G$}
	\For{ $\variable{u} \in \Call{postorder}{\variable{T}}$ }
		\State $\variable{S} \gets \variable{T}.\Call{successors}{\variable{u}}$
		\State $\mathsf{Factors}.\Call{append}{\Call{MakeRooted}{\variable{H}[\variable{S} \cup \{\variable{u}\}], \variable{u} }} $
		\State $\variable{H}.\Call{contract}{\variable{S}, \variable{u}}$ \Comment{contract all of $S$ into $u$}
	\EndFor
	\State \Return $\variable{Factors}$
\EndFunction
\State {}
\Function{AllFactors}{\variable{G}}
	\State $\variable{Factors} \gets []$
	\ForAll{ $\variable{G1} \in \Call{SCCFactors}{\variable{G}}$ }
		\ForAll{ $\variable{G2} \in \Call{DominationFactors}{\variable{G1}}$ }
			\ForAll{ $\variable{G3} \in \Call{SCCFactors}{\variable{G2}}$ }
				\State $\variable{Factors}.\Call{append}{\variable{G3}}$	
			\EndFor
		\EndFor
	\EndFor
	\State \Return $\variable{Factors}$
\EndFunction
\end{algorithmic}
\end{algorithm}
\newpage

\section{Omitted proofs}
\label{sec:omitted}

\subsection{Proofs from Section~\ref{sec:primality}}

For a directed path $\mathcal{P}$ and vertices $x,y \in \mathcal{P}$
appearing on $\mathcal{P}$ in this order, we define a \emph{slice}
$\mathcal{P}[x:y]$ as the sub-path of $P$ having $x$ as a starting vertex and
$y$ as an ending vertex. Similarly, for a directed cycle $\mathcal{C}$ and $u,v
\in \mathcal{C}$, we define a slice $\mathcal{C}[u:v]$ as the directed path from
$u$ to $v$ using the edges of the cycle $C$.

\restatablegraphone*
\begin{proof}
Strong connectivity of $G \setminus \{v\}$ implies the existence of a simple
directed path from $u$ to $w$. This path and the edge $wu$ form a directed cycle
$\mathcal{C'}$ containing edge $wu$.
By Menger's Theorem, there are two directed vertex-disjoint paths
$\mathcal{P}_1,\mathcal{P}_2$ from the root $v$ to vertex $w$. 
The paths are not necessarily disjoint with $\mathcal{C'}$. 
Let us denote, in reverse order on the cycle $\mathcal{C'}$, starting from $w$,
all intersections $c_1,c_2,\ldots, c_k$ of $\mathcal{C'}$ with $\mathcal{P}_1$
or $\mathcal{P}_2$.
Assume, without loss of generality, that $c_k \in \mathcal{P}_1$. We construct
the cycle $\mathcal{C} = \mathcal{P}_1 [c_k : w] \cup \mathcal{C'} [w : c_k]$.
Observe that $\mathcal{P}_2$ is vertex-disjoint with $\mathcal{P}_1$, and with
$\mathcal{C'}[w : c_k]$, since all intersections with $\mathcal{C}$ are between
$c_k$ and $w$. Thus, $\mathcal{P}_2$ is an independent path for $w$ with
respect to $\mathcal{C}$. Similarly, $\mathcal{P}_1[v : c_k]$ is an independent
path for $c_k$ with respect to $\mathcal{C}$.
\end{proof}

\restatablegraphtwo*
\begin{proof}
Since $G \setminus \{v\}$ is strongly connected, there exists an edge $wu$ with
$u \in V_1$ and $w \in V_2$.
By Lemma~\ref{lem:graph1}, there is a cycle $\mathcal{C}$ and a vertex $x$ such
that $x$ and $w$ are independent from $\mathcal{C}$. 
Without loss of generality, we can assume that there are no independent vertices
on $\mathcal{C}$ between $w$ and $x$ (as we can always pick $x$ to be closer to
$w$).
If $x \in V_1$, $C$ satisfies the statement of the Lemma and we are done. Assume
therefore that $x \in V_2$ from now on.

Let $\mathcal{P}_1$ and $\mathcal{P}_2$ be the independent paths of $w$ and $x$, respectively. Since
$w$ is not dominating $u$, there exists a $v$-$u$ path $\mathcal{P}$ omitting
$w$.
We can assume that $\mathcal{P}$ intersects $\mathcal{C}$, as otherwise
$\mathcal{P}$ is an independent path for $u$, and the Lemma follows.
Moreover, $\mathcal{P}$ must intersect $\mathcal{C}[x:w]$, as otherwise the
first intersection $z$ of $\mathcal{P}$ with $\mathcal{C}$ is an independent
vertex (with $\mathcal{P}[v:z]$ its independent path) that lies between $w$ and
$x$, a contradiction.

Let $y$ be the last intersection of $\mathcal{P}$ with $\mathcal{C}[x:w]$.
We show that after $y$, $\mathcal{P}$ does not intersect $\mathcal{P}_1$ nor
$\mathcal{P}_2$.
Assume for contradiction that $\mathcal{P}[y:u]$ intersects path
$\mathcal{P}_1$ or $\mathcal{P}_2$, and let $y'$ be the last such intersection.
Let $\bar{u}$ be the first intersection of $\mathcal{P}[y':u]$ with
$\mathcal{C}$. Observe that $\bar{u}$ always exists (it is, at the latest, the
vertex $u$), and that, by our definition of $y$ and $y'$, $\bar{u}\in
\mathcal{C}[w:x]\setminus\{x\}$.
Let $i\in\{1,2\}$ be the index of the path such that $y' \in \mathcal{P}_i$.
Observe that $\mathcal{P}_i[v:y'] \cup \mathcal{P}[y':\bar{u}]$ is an
independent path from $\mathcal{C}$ and thus $\bar{u}$ is an independent vertex
between $w$ and $x$ (note that $\bar{u}\neq w$), a contradiction.

Let $s_0 = u, s_1, s_2, \ldots, s_k$ be all the intersection vertices of
$\mathcal{P}[y:u]$ and $\mathcal{C}[w:x]\setminus\{w,x\}$.
We now iteratively define the sequence $c_0,c_1,c_2,\ldots,c_{2d}$ as follows.
We set $c_0 = u$ and $c_1=s_1$. Then, for every $\ell=1,2,\ldots$, we do: If
$s_k\in \mathcal{C}[w:c_{2\ell-1}]$, then we set $c_{2\ell} = s_k$ and terminate
the creation of the sequence. Else, let $s_i$ be the value of $c_{2\ell-1}$, and
we set $c_{2\ell} = s_j$, where $j$ is the largest $j\geq i$ such that $s_j\in
\mathcal{C}[w:s_i]$. Note that it is possible that $c_{2\ell} = c_{2\ell-1}$. We
also set $c_{2\ell+1} = s_j+1$.
Observe that $c_{2d} = s_k$. Now, one of the following two cases happen:
\begin{description}
\item[$c_{2d} \in V_1$:] Consider a cycle $\mathcal{C}' =
\mathcal{C}[c_{2d}:y] \cup \mathcal{P}[y:c_{2d}]$. Observe that
$\mathcal{C}[w:c_{2d}] \cup \mathcal{P}_1$ is an independent path for $c_{2d}$,
and $\mathcal{P}_2$ is an independent path for $x$ (with respect to
$\mathcal{C}'$), and the Lemma follows. 
\item[$c_{2i} \in V_2$ and $c_{2i-2} \in V_1$ for some $i$:]
Consider a cycle $\mathcal{C}'' = \mathcal{P}[c_{2i-1} : c_{2i-2}] \cup
\mathcal{C} [c_{2i-2} : c_{2i-1}]$. Observe that $\mathcal{C}[w : c_{2i-2}] \cup
\mathcal{P}_1$ is an independent path for $c_{2i-2}$, and $\mathcal{P}[c_{2i+1}
: c_{2i}] \cup \mathcal{C}[c_{2i+2} : c_{2i+1}] \cup \mathcal{P}[c_{2i+3} :
c_{2i+2}] \cup \ldots \cup \mathcal{C}[c_{2d} : c_{2d-1}] \cup \mathcal{P}[y :
c_{2d}] \cup \mathcal{C}[x : y] \cup \mathcal{P}_2$ is an independent path for
$c_{2i}$ (with respect to $\mathcal{C}''$), and the Lemma follows. 
\end{description}

\end{proof}

\restatablenofacttwo*
\begin{proof}
As in the proof of Theorem~\ref{th:nonfact1}, assume that $\kappa(G) = P \cdot
Q$, and $P$ and $Q$ are nontrivial factors. Let $V_1$ and $V_2$ be the set of
vertices with incoming edges in $\var(P)$ and $\var(Q)$, respectively. Due to
Lemma~\ref{lem:ver_part}, $V_1$ and $V_2$ define a partition of $V \setminus
\{v\}$.

Let us choose $\mathcal{C}$ and two vertices $x,y$ according to the statement of
Lemma~\ref{lem:graph2}, with $\mathcal{P}_x$ being an independent path of a
vertex $x \in V_1$, and $\mathcal{P}_y$ being an independent path of a vertex $y
\in V_2$. Also, let $e_x$ and $e_y$ be the edges from $\mathcal{C}$ ending at
$x$ and $y$, respectively.

Thus, $(\mathcal{C} \setminus \{e_x\}) \cup \mathcal{P}_x$ is a simple path from
$v$ to $x$ and there exists an arborescence $A_x$ containing this path
as a subdigraph. Respectively, there exists an arborescence $A_y$ with path
$(\mathcal{C} \setminus \{e_y\}) \cup \mathcal{P}_y$ as a subdigraph.
Thus, the polynomial $P$ has to contain a monomial corresponding to $A_x[V_1]$
and $A_y[V_1]$, and the polynomial $Q$ has to contain a monomial corresponding
to $A_x[V_2]$ and $A_y[V_2]$.
Thus, there exists a monomial in the product of $P$ and $Q$ which corresponds to
the set of edges $A = A_y[V_1] \cup A_x[V_2]$. However, $\mathcal{C} \subseteq
A$, and that contradicts with the definition of arborescence, since $A$ has
to be acyclic.
\end{proof}

\subsection{Proofs from Section~\ref{sec:decomposition}}

We start with the following two observations that relate the unrooted enumerator $\kappa(G)$ and its rooted version $\kappa_v(G)$:

\begin{observation}
$$\kappa(G) = \sum_{v \in V} \kappa_v(G)\qquad\text{and}\qquad \kappa_v(G) = \kappa(\rooot_v(G)).$$
\end{observation}

The next observation comes from the fact that in order to construct an arborescence in a digraph, it is necessary and sufficient to pick exactly one incoming edge for each vertex except the root.

\begin{observation}
\label{obs:acyclic}
For an arbitrary $G$,  $v \in V(G)$, and $\{v_1,v_2,\ldots,v_{n-1}\} = V(G) \setminus \{v\}$:
\begin{equation}
\label{eq:kappa2}
\kappa_v(G) =\!\!\!\!\sum_{\substack{e_1 \in \out(v_1),\\ \ldots,\\e_{n-1} \in \out(v_{n-1})}} [\{e_1,\ldots,e_{n-1}\} \text{ is acyclic}] \cdot \lab{e_1} \cdot \ldots \cdot \lab{e_{n-1}},
\end{equation}
\end{observation}

where $[.]$ denotes the characteristic function.


\restatablesccdecomposition*
\begin{proof}
W.l.o.g. we can assume that $G[V_1]$ is the only initial component of $G$. 
Thus $\kappa(G) = \sum_{v \in V_1} \kappa_v(G)$ and it is enough to show that for any $v \in V_1$:
\begin{equation}
\label{eq:kappa1}
\kappa_v(G) = \kappa_v(\comp{G}{V_1}) \cdot \kappa(\comp{G}{V_2}) \cdot\ldots\cdot \kappa(\comp{G}{V_k}).
\end{equation}

Let us set $V_i = \{v^i_1,v^i_2,\ldots,v^i_{n_i}\}$ for $i = 2, \ldots, k$ and $V_1 = \{v^1_1, v^1_2, \ldots, v^1_{n_1-1}\} \cup \{v\}$.
Additionally, we denote by $u_i$ the auxiliary vertex of $\comp{G}{V_i}$ (for $i=2,\ldots, k$). Observe that for $i=2,\ldots, k$ each $\comp{G}{V_i}$ is rooted at $u_i$. Thus we can rewrite the right side of \eqref{eq:kappa1} as\footnote{For any $x \in V_i$, $\comp{G}{V_i}$ preserves the labels of its incoming edges from the rest of $G$, thus we can safely write $\out(x)$ without explicitly specifying the digraph.}:
\begin{align} 
\kappa&_v(\comp{G}{V_1}) \cdot  \kappa(\comp{G}{V_2}) \cdot \ldots \cdot \kappa(\comp{G}{V_k}) = \kappa_v(\comp{G}{V_1}) \cdot \kappa_{u_2}(\comp{G}{V_2})\cdot \ldots \cdot \kappa_{u_k}(\comp{G}{V_k}) \stackrel{\mbox{\small{Obs. \ref{obs:acyclic}}}}{=} \nonumber\\
\stackrel{\mbox{Obs. \small{\ref{obs:acyclic}}}}{=}&
\left(\sum_{\substack{e^1_1 \in \out(v^1_1),\\\ldots,\\e^1_{n_1-1} \in \out(v^1_{n_1-1})}} 
[\{e^1_1,\ldots,e^1_{n_1-1}\} \text{ is acyclic}] \cdot \lab{e^1_1}\cdot\ldots \cdot \lab{e^1_{n_1-1}}
\right) \cdot \nonumber\\
&\cdot \left(\sum_{\substack{e^2_1 \in \out(v^2_1),\\\ldots,\\e^2_{n_2} \in \out(v^2_{n_2})}} 
[\{e^2_1,\ldots,e^2_{n_2}\} \text{ is acyclic}] \cdot \lab{e^2_1}\cdot\ldots\cdot \lab{e^2_{n_2}}
\right) \cdot \ldots \nonumber\\
\ldots &\cdot \left(\sum_{\substack{e^k_1 \in \out(v^k_1),\\\ldots,\\e^k_{n_k} \in \out(v^k_{n_k})}} 
[\{e^k_1,\ldots,e^k_{n_k}\} \text{ is acyclic}] \cdot \lab{e^k_1}\cdot\ldots \cdot \lab{e^k_{n_k}}
\right) = \nonumber
\end{align}
\begin{align}
\label{eq:kappa3}
=\sum_{\substack{e^1_1 \in \out(v^1_1),\\\ldots,\\e^1_{n_1-1} \in \out(v^1_{n_1-1})}}  
\!\!\!\ldots\!\!\!\! \sum_{\substack{e^k_1 \in \out(v^k_1),\\\ldots,\\e^k_{n_k} \in \out(v^k_{n_k})}} \!\!\!\!&\Big( [\{e^1_1,\ldots,e^1_{n_1-1}\} \text{ is acyclic}] \cdot \ldots \cdot [\{e^k_1,\ldots,e^k_{n_k}\} \text{ is acyclic}] \cdot \nonumber\\
&\cdot \left( \lab{e^1_1} \cdot \ldots \cdot \lab{e^1_{n_1-1}} \right)\cdot \ldots \cdot \left( \lab{e^k_1} \cdot  \ldots \cdot \lab{e^k_{i_k}} \right) \Big).
\end{align}

Since $\{v_1,v_2,\ldots,v_{n-1}\} = V(G) \setminus \{v\} = (V_1 \setminus\{v\}) \cup V_2 \cup \ldots \cup V_k = \{v^1_1,\ldots,v^1_{n_1-1},\ldots,v^k_1,\ldots,v^k_{n_k}\}$, both \eqref{eq:kappa2} and \eqref{eq:kappa3} 
are summing over the same ranges (in a permuted order). Thus it is enough to prove that
$$\{e_1,\ldots,e_{n-1}\} \text{ is acyclic in }G$$
$$ \text{ iff }$$
$$(\{e^1_1,\ldots,e^1_{n_1-1}\} \text{ is acyclic in }G[V_1]) \text{ and } \ldots \text{ and }  (\{e^k_1,\ldots,e^k_{n_k}\} \text{ is acyclic in }G[V_k]),$$
where $\{e^1_1,\ldots,e^1_{n_1-1}\},\ldots,\{e^k_1,\ldots,e^k_{n_k}\}$ is the partitioning of $\{e_1,e_2,\ldots,e_{n-1}\}$ by the strongly connected component to which the target vertex belongs.
However, it is enough to notice that any cycle in $G$ can only span vertices from a single strongly connected component.
\end{proof}

\restatabledominationdecomposition*
\begin{proof}
Observe that \eqref{eq:kappa4} is equivalent to:
$$\kappa_v(G) = \kappa_u(G[D]) \cdot \kappa_v( G(D \to u) ).$$

Let us denote $D \setminus \{u\} = \{v_1,v_2,\ldots,v_{i-1}\}$, and let $V \setminus D \setminus \{v\} = \{v_{i+1},\ldots,v_{n-1}\}$. Thus, $V = \{v_1,\ldots,v_{i-1},u,v_{i+1},\ldots,v_{n-1},v\}$.

Observe, that in $G$ there are no edges going from any vertex from $V \setminus D$ to any vertex from $D \setminus \{ u \}$, as otherwise $u$ would not dominate said vertices.
If we denote for short $\out_1(u) = \out_G(u)$ (all incoming edges to $u$) and $\out_2(u) = \out_{G(D \to u)}(u)$ (edges incoming from $V\setminus D$), we have that $\out_2(u) \subseteq \out_1(u)$. Thus:
$$ \kappa_v(G(D \to u)) = 
\sum_{e_i \in \out_2(u)}\sum_{\substack{\\e_{i+1} \in \out(v_{i+1}),\\\ldots,\\e_{n-1} \in \out(v_{n-1})}} [\{e_i,\ldots,e_{n-1}\} \text{ is acyclic}] \cdot \lab{e_i}\ldots \lab{e_{n-1}}.
$$

Observe that any arborescence of $G$ rooted at $v$ cannot use an edge from $\out_1(u) \setminus \out_2(u)$, as that would create disconnected digraph  (as any path going from $v$ to $D$ has to go through $u$). Thus:
$$ \kappa_v(G) = 
\sum_{\substack{e_1 \in \out(v_1),\\\ldots,\\e_{i-1} \in \out(v_{i-1})}}
\sum_{e_i \in \out_2(u)}
\sum_{\substack{e_{i+1} \in \out(v_{i+1})\\\ldots,\\e_{n-1} \in \out(v_{n-1})}} [\{e_1,\ldots,e_{n-1}\} \text{ is acyclic}] \cdot \lab{e_1}\ldots \lab{e_{n-1}}.
$$

Additionally, in $G[D]$, all vertices from  $D \setminus \{u\}$ have the same incoming edges as in $G$ (as taking induced subdigraph removes edges coming from outside, but there were no such edges).
 Thus by Observation~\ref{obs:acyclic}:
$$ \kappa_u(G[D]) = 
\sum_{\substack{e_1 \in \out(v_1),\\\ldots\\e_{i-1} \in \out(v_{i-1})}} [\{e_1,\ldots,e_{i-1}\} \text{ is acyclic}] \cdot \lab{e_1}\cdot \ldots \cdot \lab{e_{i-1}}.
$$

Thus, it is enough to prove that, for $e_i \in \out_2(u), e_1 \in \out(v_1),\ldots,e_{i-1} \in \out(v_{i-1}), e_{i+1} \in \out(v_{i+1}), \ldots, e_{n-1} \in \out(v_{n-1})$:
$$\{e_1,\ldots,e_{n-1}\} \text{ is acyclic in }G$$
$$\text{iff}$$
$$(\{e_1,\ldots,e_{i-1}\} \text{ is acyclic in }\rooot_u(G[D]))\text{ and } (\{e_i,\ldots,e_{n-1}\} \text{ is acyclic in }G(D \to u))$$

To prove it in one direction, observe that any cycle in $G[D]$ remains cycle in $G$. Similarly, any cycle in $G(D \to u)$ remains a cycle in $G$. To prove it in another direction, observe that any cycle in $G$ (with constrains on $e_i \in \out_2(u)$) either spans vertices only from $D \setminus \{u\}$ thus remains a cycle in $G[D]$, or spans at least one vertex from $G \setminus (D \setminus \{u\})$. In the latter case, the digraph remains cyclic when contracting whole $D$ into $u$.
\end{proof}

\restatableththreelevels*
\begin{proof}
Let $H$ be a digraph obtained through the application of rule \eqref{eq:tree_decomp}, that is
$$H = \rooot_u (G( \dom{G}{u_1} \to u_1; \ldots ;\dom{G}{u_i}\to u_i)[\{u,u_1,\ldots,u_i\}]).$$

First, observe that in digraph $F = \rooot_u( G[ \dom{G}{u} ] )$ the domination relation between vertices (with respect to vertex $u$) is the same as is in $G$ (with respect to vertex $v$). Thus, as we have that
$$H = F( \dom{G}{u_1} \to u_1; \ldots ;\dom{G}{u_i}\to u_i),$$
there are no non-trivial dominance relations in $H$.

What is now left to prove is that for any strongly connected component $V_j \not= \{u\}$ of $H$, the digraph $\comp{H}{V_j}$ has no non-trivial dominators (with respect to $v_\aux$, its root). However, for any vertex $w \in V_j$ and for any path going from $u$ to $w$ in $H$, there is a path from $v_\aux$ to $w$ in $\comp{H}{V_j}$ corresponding to the suffix of the former one. Thus, if there is any vertex $w'$ dominating (non-trivially) $w$ in $\comp{H}{V_j}$, then $w'$ dominates (non-trivially) $w$ in $H$, a contradiction.
\end{proof}

\newpage
\section{Experimentally evaluated digraphs}
\label{app:expgraphs}
\enlargethispage{5cm}
Description and references of the experimentally evaluated digraphs in Section~\ref{sec:applications} are shown in Table~\ref{tab:netref}.
Note that in some application domains, including the Laplacian framework, rooted directed spanning trees for which all edges point towards the root, called in-arborescences, have to be enumerated. This is not limiting since there is a bijection between the in-arborescences of a digraph $G$ and the (out-) arborescences of the transpose digraph of $G$. 

\begin{table}[h!]
  \begin{center}
  \caption{Digraphs to which the recursive enumeration algorithm was applied. Shown are the digraph aliases, which type of arborescences were enumerated, and a short description for each digraph. }
\label{tab:netref}
    \begin{tabular}{l l p{.65\linewidth}}
\hline\noalign{\smallskip}
	Alias& Enumerated& Description\\
\noalign{\smallskip}
\hline
\noalign{\smallskip}
ColE1&		in-arborescences&		Kinetic scheme of the ColE1 replication control mechanism from	\cite{shin2000effects}.\\
\noalign{\smallskip}
Brain&		out-arborescences&		Brain effective connectivity model from \cite{xu2014pooling}.\\
\noalign{\smallskip}
MDH&		in-arborescences&		Proposed kinetic mechanism for the reaction cycle of \textit{M.methylotrophus} methanol dehydrogenase (MDH) with ammonium as activator from \cite{hothi2005kinetic}.\\
\noalign{\smallskip}
KNF33&		in-arborescences&		Three by three classical KNF (Koshland, Nemethy, and Filmer) model for allosteric regulation of enzymes.\\
\noalign{\smallskip}
PHO5&		in-arborescences&		Regulation of yeast PHO5 gene from \cite{ahsendorf2014framework}.\\
\noalign{\smallskip}
GR&			in-arborescences&		Scheme for the catalytic mechanism of glutathione reductase (GR) from \cite{pannala2013biophysically}.\\
\noalign{\smallskip}
ERG&		in-arborescences&		Randomly generated Erd\H{o}s–R\'enyi digraph with probability for edge creation $p=0.2$.\\
\noalign{\smallskip}
SF1&		out-arborescences&		The largest SCC of a randomly generated scale-free digraph with $n=1000$, probability for connecting a new node to an existing one chosen randomly according to in-degree distribution $\alpha=0.4$, probability for adding an edge between two existing nodes $\beta=0.4$, and probability for connecting a new node to an existing one chosen randomly according to out-degree distribution $\gamma=0.2$.\\
\noalign{\smallskip}
SF3&		out-arborescences&		As above, with $n=3000$.\\
\noalign{\smallskip}
S838&		in-arborescences&		An electronic benchmark circuit from \cite{s838}.\\
\noalign{\smallskip}
PE45&		in-arborescences&		Randomly generated practically enumerable digraph from prime factors with topology $\lbrace(1,3),(1,2),(2,3),(3,2)\rbrace$, with recursion depth $d = 4$, and number of prime factors per recursion level $n = 5$.\\
    \hline
    \end{tabular}
  \end{center}
\end{table}

\end{document}